\documentclass[11pt]{amsart}
\usepackage{amsmath}
\usepackage{amssymb}
\usepackage{amsfonts}
\usepackage{latexsym}
\usepackage{fancyhdr}
\usepackage{graphics}
\usepackage{color}
\newcommand{\norm}[1]{ \parallel #1 \parallel}
\newcommand{\mb}{\mathbb}
\newcommand{\mc}{\mathcal}
\newcommand{\eul}{\mathfrak}
\newcommand{\A}{\eul A}
\newcommand{\Ao}{{\eul A}_{\scriptscriptstyle 0}}
\newcommand{\M}{\eul M}
\newcommand{\D}{\mc D}
\newcommand{\HH}{\mc H}

\newcommand{\id}{{\mb I}}

\newcommand{\hh}{\mc H}

\newcommand{\X}{{\mathfrak X}}

\newcommand{\vp}{\varphi}

\newtheorem{defn}{Definition}[section]
\newtheorem{prop}[defn]{Proposition}
\newtheorem{thm}[defn]{Theorem}
\newtheorem{lemma}[defn]{Lemma}

\newtheorem{example}[defn]{Example}
\newtheorem{rem}[defn]{Remark}
\def\x{\relax\ifmmode {\mbox{*}}\else*\fi}
\newcommand{\beex}{\begin{example}$\!\!${\bf }$\;$\rm }
\newcommand{\enex}{ \end{example}}
\newcommand{\berem}{\begin{rem}$\!\!${\bf }$\;$\rm }
\newcommand{\enrem}{ \end{rem}}
\newcommand{\aff}{\,\eta\,}
\newcommand{\BH}{\mc{B}(\mathcal{H})}

\newcommand{\bedefi}{\begin{defn}$\!\!${\bf }$\;$\rm }
\newcommand{\findefi}{\end{defn}}
\begin{document}
\title[Quasi C*-algebras and noncommutative integration]{Locally convex quasi C*-algebras and noncommutative integration}%

\author{Camillo Trapani}%
\address{Dipartimento di Matematica e Informatica, Universit\`a
di Palermo, I-90123 Palermo (Italy)} \email{camillo.trapani@unipa.it}
\author{Salvatore Triolo}%
\address{Dipartimento DEIM, Universit\`a di Palermo, I-90123 Palermo (Italy)}%
\email{salvatore.triolo@unipa.it} \subjclass[2000]{Primary 46L08;
Secondary 46L51, 47L60 }

\keywords{Banach C*-modules, Noncommutative integration, Partial algebras of operators}

\begin{abstract} In this paper we continue the analysis undertaken in a series of previous papers on structures arising as completions of C*-algebras under topologies coarser that their norm and we focus our attention on the
so-called {\em locally convex quasi C*-algebras}. We show, in particular, that any strongly
*-semisimple locally convex quasi C*-algebra $(\X,\Ao)$, can be represented
in a class of noncommutative local $L^2$-spaces.
\end{abstract}
\maketitle

\section{Introduction }

The completion $\X$ of a C*-algebra $\Ao$, with respect to a norm
weaker than the C*-norm  provides a mathematical framework where discussing certain quantum physical systems for which the usual algebraic approach made in terms of C*-algebras revealed to be insufficient.

First of all, $\X$ is a Banach $\Ao$-module and becomes a
quasi *-algebra if $\X$ carries an involution which extends the
involution $*$ of $\Ao$. This structure has been called a {\it
proper CQ*-algebra} in a series of papers
\cite{standard}-\cite{bagtratri}, \cite{4}-\cite{cqisaac} to that we refer for a detailed
analysis. On the other hand, if $\X$ is endowed with an isometric
involution different from that of $\Ao$, then the structure
becomes more involved.

 CQ*-algebras are examples of  more general structures called {\em locally convex quasi C*-algebras} \cite{con maria}.
They are obtained by completing a  C*-algebra with respect to a new locally convex topology $\tau$ on $\A_0$ compatible with the corresponding $\| \cdot \|$-topology.
Under certain conditions on $\tau$ a quasi $*$-subalgebra $\A$ of the completion $\widetilde{\A_0}[\tau]$ is a locally convex quasi *-algebra which is named locally convex quasi C*-algebra.

In \cite{bttjmp} quasi *-algebras of measurable and/or integrable operators (in the sense of Segal \cite{segal}, \cite{S.Triolocirc2} and Nelson \cite{nelson}) were examined in detail and it was proved that any *-semisimple CQ*-algebra can be realized as a CQ*-algebra of
measurable operators, with the help of a particular class of positive bounded
sesquilinear forms on $\X$.

In this paper, after a short overview of the main results obtained on this subject,
we continue our study of locally convex quasi C*-algebras and
we generalize the results obtained in \cite{bttjmp}  for proper CQ*- algebras to these structures.

The main question we pose now in the present paper is the following: given a
*-semisimple locally convex quasi C*-algebras $(\X, \Ao)$ and
the universal *-repre\-sen\-ta\-tion of $\Ao$, defined via the
Gelfand-Naimark theorem, can $\X$ be realized as a locally convex quasi C*-algebra of
operators  of type  $L^2$?

The paper is organized as follows. We begin with giving a short overview of noncommutative $L^p$-spaces (constructed starting from
a von Neumann algebra $\M$ and a normal, semifinite, faithful trace $\varphi$ on $\M$), considered as CQ*-algebras.  We also introduce the noncommutative $L^p_{\rm loc}$-space constructed on a von Neumann algebra possessing a family of mutually orthogonal central projections whose sum is the identity operator. We show that $(L^p_{\rm loc}(\varphi), \M)$ is a locally convex quasi C*-algebra.

Finally we give some results on the structure of locally convex quasi C*-algebras: we prove that any locally convex quasi C*-algebra $(\X,\Ao)$
possessing a sufficient family of bounded positive tracial
sesquilinear forms  can be continuously embedded into a
locally convex quasi C*-algebras of measurable operators
of the type $(L^2_{\rm loc}(\varphi), \M).$

\subsection{Definitions and results on noncommutative
measure}

\bigskip
The following basic definitions and results on noncommutative
measure theory and integration are needed in what follows.
 Let $\M$ be a von Neumann
algebra on a Hilbert space $\HH$ and $\vp$ a normal faithful semifinite trace defined on
$\M_+$.

Put $${\mc J}=\{ X \in \M: \vp(|X|)<\infty \}.$$ ${\mc J}$ is a
*-ideal of $\M$.

Let $P\in {\rm Proj}{(\M)}$, the lattice of projections of $\M$.
Two projections $P,Q\in {\rm Proj}{(\M)}$ are called equivalent,
$P \sim Q $ if there is a $U\in \M$ with $U^*U=P$ and $UU^*=Q.$ We
say that $P\prec Q$ in the case where $P$ is equivalent to a
subprojection of $Q.$

\medskip
A projection $P$ of a von Neumann algebra $\M$ is said to be {\em
finite} if $P\sim Q \leq P$ implies $P=Q.$ A projection $P\in
\M$ is said to be {\em purely infinite} if there is no nonzero
finite projection  $Q\preceq P$ in $\M.$ A von Neumann algebra
$\M$ is said to be {\em finite} (respectively, {\em purely
infinite}) if the identity operator $\id$ is finite (respectively,
purely infinite).

We say that $P$ is $\vp$-finite if $P \in {\mc J}$. Any
$\vp$-finite projection is finite.

In what follows we will need the following result (see, \cite[Vol. IV, Ex. 6.9.12]{kadison}) that we state as a lemma.
\begin{lemma} \label{risultato k} Let $\M$ be a von Neumann algebra  on a Hilbert space $\HH$ and $\vp$ a normal faithful semifinite trace defined on
$\M_+$. There is an orthogonal family $\{Q_j; {j\in J} \}$  of nonzero central projections
in $\M$ such that $ \bigvee_{j\in J} Q_j =\id $ and each $Q_j$ is the sum of an orthogonal
family of mutually equivalent finite projections in $\M$  .
\end{lemma}

A vector subspace $\D$ of $\HH$ is said to be {\em strongly
dense} ( resp., {\em strongly $\vp$-dense}) if
\begin{itemize}
\item $U'\D \subset \D$ for any unitary $U'$ in $\M'$;
\item there exists a sequence $P_n \in {\rm Proj}{(\M)}$: $P_n \HH
\subset \D$ , $P_n^\perp \downarrow 0$ and ${P_n^\perp}$ is a
finite projection (resp., $\vp{(P_n^\perp)}<\infty$).
\end{itemize}

Clearly, every strongly $\vp$-dense domain is strongly dense.

Throughout this paper, when we say that an operator $T$ is affiliated with a von Neumann algebra, written $T \aff \M$, we always mean that $T$ is closed, densely defined and $TU\supseteq UT$ for every unitary operator $U \in \M'$.

An operator $T \aff \M$ is called
\begin{itemize}
\item {\em measurable} (with respect to
$\M$) if its domain $D(T)$ is strongly dense;
\item {\em $\vp$-measurable} if its domain
$D(T)$ is strongly $\vp$-dense. \label{meaurability}
\end{itemize}

From the definition itself it follows that, if $T$ is
$\vp$-measurable, then there exists $P \in {\rm Proj}{(\M)}$ such
that $TP$ is bounded and $\vp(P^\perp)<\infty$. \label{main}

We remind that any operator affiliated with a finite von Neumann
algebra is measurable \cite[Cor. 4.1]{segal} but it is not
necessarily $\vp$-measurable.

\berem The following statements will be used later.
\begin{itemize}
\item[(i)] Let $T \aff \M$ and $Q \in \M$. If $D(TQ)=\{\xi \in \HH: Q\xi \in
D(T)\}$ is dense in $\HH$, then $TQ \aff \M$. \label{tq2}
\item[(ii)] If $Q\in {\rm Proj}{(\M)}$, then $Q\M Q=\{ QXQ\upharpoonright_{Q\HH}; X
\in \M\}$ is a von Neumann algebra on the Hilbert space $Q\HH$;
moreover $(Q\M Q)'=Q\M'Q$.  If $T \aff \M$ and $Q \in \M$ and
$D(TQ)=\{\xi \in \HH: Q\xi \in D(T)\}$  is dense in $\HH$, then
$QTQ \aff Q\M Q$.
\end{itemize}
\enrem

Let $\M$ be a von Neumann algebra  on a Hilbert space $\HH$ and $\vp$ a normal faithful
semifinite trace defined on $\M_+$. For each $p\geq 1$, let
$${\mc J_p}=\{ X \in \M: \vp(|X|^p)<\infty \}.$$ Then ${\mc J_p}$ is a
*-ideal of $\M$. Following \cite{nelson}, we denote with
$L^p(\vp)$ the Banach space completion of ${\mc J_p}$ with respect
to the norm
$$\|X\|_{ {p, \vp}} := \vp(|X|^p)^{1/p}, \quad X \in {\mc J_p}.$$
One usually defines $L^\infty(\vp): = \M$. Thus, if $\vp$ is a
finite trace, then $L^\infty(\vp) \subset L^p(\vp)$ for every
$p\geq 1$. As shown in \cite{nelson}, if $X \in L^p(\vp)$, then
$X$ is a measurable operator.

If $A$ is a {measurable} operator and $A\geq 0$, one defines the {\em integral} of $A$ by
$$ \mu (A) = \sup\{\vp(X):\; 0 \leq X \leq A, \, X \in {\mc J}_1 \}.$$

Then the space $L^p(\vp)$ can also be defined \cite{nelson}  as
the space of all measurable operators $A$ such that $\mu (|A|^p) < \infty
$.

The {integral} of an element $A\in L^p(\vp)$ can be defined,
in obvious way, taking into account that any measurable operator
$A$ can be decomposed as $A= B_+-B_- +iC_+-iC_-$, where
$B=\frac{A+A^*}{2}$ and $C=\frac{A-A^*}{2i}$ and $B_+, B_-$ (resp. $C_+, C_-$) are the positive and negative parts of $B$ (resp. $C$).

\subsection{Locally convex quasi C*-algebras}

\bigskip

\vspace{1ex} In what follows we recall some definitions and facts
needed in the sequel.
 \bedefi Let $\X$ be a complex vector space and $\Ao$ a  $^\ast$ -algebra contained in $\X$. $\X$ is said
  a \textit{quasi  $^\ast$-algebra with distinguished  $^\ast$-algebra $\Ao$} (or, simply, over $\Ao$) if

\begin{itemize}\item[(i)]\label{11} the multiplication of $\A_0$ is extended on $\X$ as
follows: the correspondences

$\X \times \A_0 \rightarrow \A : (a,x) \mapsto ax \text{ (left
multiplication of $x$ by $a$) } $ and

$ \A_0 \times \X \rightarrow \A: (x,a) \mapsto xa
\text{ (right multiplication of $x$ by $a$) }$\\
are always defined and are bilinear;

\item[(ii)] $x_1(x_2 a) = (x_1 x_2) a, (a x_1)x_2= a(x_1 x_2)$ and $x_1
(a x_2) = (x_1 a ) x_2$, for all $x_1, x_2 \in \A_0$ and $a \in
\X$;

\item[(iii)] the involution $*$ of $\A_0$ is extended on $\X$, denoted
also by $*$, and satisfies $(ax)^* = x^* a^*$ and $(xa)^* = a^* x^*$,
for all $x \in \A_0$ and $a \in \X$.
  \end{itemize}
 \findefi

\bigskip Thus a {\em quasi *-algebra} \cite{schbook} is a
couple $(\X, \A_{\scriptscriptstyle 0})$, where $\X$ is a vector space with involution
$^*$, $\A_{\scriptscriptstyle 0}$ is a *-algebra and a vector subspace of $\X$ and $\X$
is an $\A_{\scriptscriptstyle 0}$-bimodule whose module operations and involution
extend those of $\A_{\scriptscriptstyle 0}$. The {\it unit} of $(\X, \A_{\scriptscriptstyle 0})$ is an
element $e\in \A_{\scriptscriptstyle 0}$ such that $xe=ex=x$, for every $x \in \X$.

 A quasi *-algebra $(\X, \A_{\scriptscriptstyle 0})$ is said to
be {\it locally convex} if $\X$ is endowed with a topology $\tau$
which makes $\X$ a locally convex space and such that the
involution $a \mapsto a^*$ and the multiplications $a \mapsto ab$,
$a \mapsto ba$, $b \in \A_{\scriptscriptstyle 0}$, are continuous.
If $\tau$ is a norm topology and the involution is isometric with
respect to the norm, we say that $(\X, \A_{\scriptscriptstyle 0})$
is a {\it normed quasi *-algebra}   and, if it is complete, we say
it is a {\it Banach quasi*-algebra}.

Let $\A_0[\| \cdot \|_0]$ be a C*-algebra. We shall use the
symbol $\| \cdot \|_0$ of the C*-norm to denote the corresponding
topology. Suppose that $\tau$ is a topology on $\A_0$ such that
$\A_0[\tau]$ is a locally convex $*$-algebra. Then, the topologies
$\tau$, $\| \cdot \|_0$ on $\A_0$ are  \textit{compatible} whenever
each Cauchy net in both topologies that converges with respect to
one of them, also converges  with respect to the other one.

Under certain conditions on $\tau$ a quasi $*$-subalgebra $\A$ of
the quasi $*$-algebra $\X=\widetilde{\A_0}[\tau]$ over $\A_0$ is
named {\em{locally convex quasi C*-algebra}}. More precisely, let
$\{p_\lambda\}_{\lambda \in \Lambda}$ be a directed family of
seminorms defining the topology $\tau$. Suppose that $\tau$ is compatible with $\| \cdot \|_0$ and that it has
the following properties:

\begin{itemize}
\item (T$_1$) $\A_0[\tau]$ is a locally convex $*$-algebra with separately continuous multiplication.

\item (T$_2$) $\tau \preceq \| \cdot \|_0$.
\end{itemize}

Then, the identity map $\A_0[\| \cdot \|_0] \rightarrow \A_0[\tau]$
extends to a continuous $*$-linear map $\A_0[\| \cdot \|_0]
\rightarrow \widetilde{\A_0}[\tau].$
Since $\tau, \| \cdot \|_0$
are compatible, the C*-algebra $\A_0[\| \cdot \|_0]$ can be
regarded as embedded into $\widetilde{\A_0}[\tau]$. It is easily shown
that $\widetilde{\A_0}[\tau]$ is a quasi $*$-algebra over $\A_0$
(cf. \cite[Section 3]{maria-i-kue}).

We denote by  $(\A_0)_+$  the set of all positive elements of the
C*-algebra $\A_0[\| \cdot \|_0].$

Further, we employ the following two extra conditions (T$_3$),
(T$_4$) for the locally convex topology $\tau$ on $\A_0$

\begin{itemize}
\item
 (T$_3$) For each $\lambda \in \Lambda$, there exists $\lambda' \in \Lambda$ such that
\[
p_\lambda(xy) \leq \| x \|_0 p_{\lambda'}(y), \forall \ x, y \in
\A_0 \text{ with } xy = yx;
\]
\item (T$_4$) The set $\mathcal{U}(\A_0)_+ := \{ x \in (\A_0)_+ : \| x \|_0 \leq 1 \}$ is $\tau$-closed.
\end{itemize}

\bedefi
By a {\em locally convex quasi C*-algebra} over $\A_0,$ (see \cite{con maria}), we mean any
 quasi $*$-subalgebra $\A$ of the locally convex quasi $*$-algebra
$\X=\widetilde{\A_0}[\tau]$ over $\A_0$, where $\A_0[\| \cdot \|_0]$ is
a C*-algebra with identity $e$ and $\tau$ a locally
convex topology on $\A_0$, defined by a directed family of seminorms $\{p_\lambda\}_{\lambda \in \Lambda}$,  satisfying the conditions (T$_1$)-(T$_4$).
\findefi

{  The following examples have been discussed in \cite{con maria}.}
\medskip
\beex[CQ*-algebras]
Let $\A_{\scriptscriptstyle 0}$ be a $C^\ast $-algebra, with norm $\|\cdot\|$ and
involution $^*.$ Let $\| \cdot \|_1$
 be a norm on $\A_{\scriptscriptstyle 0}$, weaker than $\| \cdot \|$
 and such that, for every $a,b \in  \A$
\begin{itemize}
\item[(i)]  $\|ab\|_1 \leq \|a\|_1 \|b\|,$
\item[(ii)]$ \|a^*\|_1=\|a\|_1.$
\end{itemize}
Let $\X$ denote the $\| \cdot \|_1 $-completion of $\A_{\scriptscriptstyle 0}$; then \footnote{In previous papers this structure was called {\em proper} CQ*-algebra. Since this is the sole case we consider here we will systematically omit the specification {\em proper}.} the couple $(\X,\A_{\scriptscriptstyle 0})$ is called a {\em CQ*-algebra}. {  Every CQ*-algebra is a locally convex quasi C*-algebra.}
\enex

\beex The space $L^p([0,1])$, with $1\leq p< +\infty$ is a Banach $L^\infty([0,1])$-bimodule. The couple
$(L^p([0,1]), L^\infty([0,1])$ may be regarded as a
CQ*-algebra thus a  locally convex quasi C*-algebra over $L^\infty([0,1])$.\enex

\section{ Locally convex quasi C*-algebras of measurable operators}

Let $\M$ be a von Neumann
algebra  on a Hilbert space $\HH$ and $\vp$ a normal faithful semifinite trace on $\M_+$,
then, as shown in \cite{bttjmp}, $(L^p(\vp), L^\infty(\vp)\cap L^p(\vp))$ is a {\it Banach quasi*-algebra}
and if $\vp$ is a finite trace,  $(L^p(\vp), \M)$ is a
 CQ*-algebra.

 In analogy to \cite{bttjmp} we consider the following two sets of sesquilinear forms
 enjoying certain invariance properties.
\bedefi \label{SLP} Let $(\X,\Ao)$ be a locally convex quasi
C*-algebra {  with unit $e$}. We denote by $\mc S(\X) $ the set of all sesquilinear
forms $\Omega$ on $\X \times \X$ with the following properties:
\begin{itemize} \label{sesqu}
\item[(i)]$ \Omega(x,x) \geq 0 \hspace{3mm} \forall x \in \X$;
\item[(ii)] $\Omega(xa,b) = \Omega (a,x^*b)  \hspace{3mm} \forall x \in
\X,\hspace{2mm} \forall a,b \in \Ao$;
\item[(iii)]$ |\Omega(x,y)| \leq p(x) p(y)$  \mbox{for some}\, $\tau$-\mbox{continuous seminorm} \, $p$ on $\X$ \mbox{and all} $x,y \in \X$;
{ \item[(iv)] $\Omega (e,e)\leq 1$.}
\end{itemize}

The locally convex quasi C*-algebra  $(\X,\Ao)$ is called {\em *-semisimple}  if $x \in
\X$,
 $\Omega(x,x)=0 $, for every $\Omega \in \mc S (\X)$,
implies $x=0$.

We denote by $\mc T(\X)\subseteq \mc S(\X) $ the set of all sesquilinear forms
$\Omega\in \mc S(\X)$ with the following property
\begin{itemize}

\item[(iv)]$ \Omega(x,x)=\Omega(x^*,x^*), \;\forall x \in \X$.

\end{itemize}

\findefi
\berem
\label{rem23}
Notice that
\begin{itemize}

\item
By (iv) of Definition \ref{SLP} and  by polarization, we get  $$\Omega(y^ {\ast},x^
{\ast})=\Omega(x,y),\quad \forall x,y \in \X.$$

\item \label{convessità} The set $
\mathcal{T}(\X)$ is convex.

\end{itemize}
\enrem

\beex \label{ex_2.3} Let $\M$ be a von Neumann
algebra and $\vp$ a normal faithful semifinite trace on $\M_+$. Then, ($ L^p(\vp) , {\mc J}_p$) , $p\geq 2$, is a $^*$-semisimple Banach quasi *-algebra.
 If $\vp$ is a finite trace { (we assume $\vp(\id)=1)$}, then $(L^p(\vp), \M)$, with $p\geq 2$, is a *-semisimple  locally convex quasi C*-algebra.
If $p\geq 2$, $L^p$-spaces possess a sufficient family of positive sesquilinear
forms.
Indeed, in this case, since for every $W\in L^p(\vp),$
$|W|^{p-2}\in L^{{p}/(p-2)}(\vp)$, then the sesquilinear form $\Omega_W$ defined by
$$\Omega_W(X,Y):=\frac{\vp[X(Y |W|^{p-2}){^\ast}]}{\norm{W}_{ {p, \vp}}^{p-2}}$$
is positive and satisfies the conditions
(i),(ii),(iii) and (iv)  of Definition \ref{SLP}, (see \cite{bttjmp},  and \cite{11} for more details).
Moreover, $$\Omega_W (W,W)=\norm{W}^{p}_{ {p, \vp}}.$$
\enex

\berem The notion of *-semisimplicity of locally convex partial *-algebras has been studied in full generality in \cite{ant_bell_ct} and \cite{10}. \enrem

\bedefi Let $\M$ be a von Neumann algebra and $\vp$ a normal faithful
semifinite trace defined on $\M_+$.
We denote by $L^p_{\rm loc}(\vp)$ the set of all measurable operators  $T$ such that  $TP\in L^p(\vp)$, for every central $\vp$-finite projection $P$ of $\M.$
 \label{msub} \findefi
\berem \label{2.6} The von Neumann algebra $\M$ is a subset of $L^p_{\rm
loc}(\vp).$ Indeed, if $X\in\M,$ then for every  $\vp$-finite
central projection $P$ of $\M$ the product $XP$  belongs to the
*-ideal ${\mc J_p}.$

\enrem

Throughout this section we are given a von Neumann algebra $\M$ on
a Hilbert space $\HH$ with a family $ \{P_j\}_{j\in J}$ of $\vp$-finite central
projections  of $\M$  such that
\begin{itemize}
\item  if $l,m \in J$, $l \neq  m$, then $P_l P_m=0$ (i.e., the $P_j$'s are orthogonal);
\item  $\bigvee_{j\in J} P_j=\id;\,$ where $\bigvee_{j\in J} P_j$ denotes the projection onto the subspace generated by $\{P_j\hh: \, j\in J \}.$
\end{itemize}
The previous two conditions are always realized in  a von Neumann  $\M$  algebra with  a faithful normal semifinite trace
(see Lemma \ref{risultato k}  and  \cite{kadison, takesaki} for more details).

If $\vp$ is a normal faithful semifinite trace on $\M_+$,  we define,
for each $X\in \M$, the following seminorms $q_j(X) :=
\|XP_j\|_{p,\vp}$ on $\M,\, j\in J.$    The traslation-invariant
locally convex topology defined by the system $\{q_j: \,j\in J\}$
is denoted by $\tau_p$.

\bedefi Let $\M$ be a von Neumann algebra and $\vp$ a normal faithful
semifinite trace defined on $\M_+$.
We denote by $\widetilde{\M}^{\tau_p}$  the $\tau_p$ completion of $\M.$
\findefi

\begin{prop}\label{prop} Let $\M$ be a von Neumann
algebra  and $\vp$ a normal faithful semifinite trace on $\M_+$.
Then $L^p_{\rm loc}(\vp)\subseteq\widetilde{\M}^{\tau_p} .$
Moreover, if there exists a family $\{P_j\}_{j\in J}$ as above, where all $P_j$' s are mutually equivalent,
then $L^p_{\rm loc}(\vp)=\widetilde{\M}^{\tau_p}.$
\end{prop}

\begin{proof}
From Remark \ref{2.6},  $\M \subseteq L^p_{\rm loc}(\vp).$ If
$Y\in L^p_{\rm loc}(\vp), $ for every ${j\in J}$ we have $YP_j\in
L^p(\vp)$. Then, for every ${j\in J}$, there exist $(X_n^j)_{n=1}^\infty\subseteq {\mc
J_p}$ such that $\displaystyle \|X_n^j- YP_j\|_{p,\vp}\underset{n\to \infty} {\longrightarrow} 0.$

{
Let ${\mb F}_J$ be the family of finite subsets of $J$ ordered by inclusion and $F\in {\mb F}_J$.

We put $$T_{n,F}:=\sum_{j\in F} X_n^jP_j \in \M. $$
Then the net $(T_{n,F})$ converges to $Y$ with respect to $\tau_p$. Indeed,
for every   ${m\in J}$,  $$ q_m(T_{n,F}-Y) = \|(T_{n,F}-Y)P_m\|_{p,\vp}=\|(X_n^m-Y)P_m\|_{p,\vp}$$ for sufficiently large $F$.
Thus, the inequality $\|(X_n^m-Y)P_m\|_{p,\vp}\leq \|X_n^m-YP_m\|_{p,\vp}$ implies that
$\displaystyle q_m(T_{n,F}-Y) \underset{n,F} {\longrightarrow} 0.$

Hence $L^p_{\rm loc}(\vp)\subseteq\widetilde{\M}^{\tau_p} .$
}

On the other hand, {  assume that all $P_j$' s are mutually equivalent}. Then, if $Y\in\widetilde{\M}^{\tau_p},$ there
exists a net  $( X_\alpha) \subseteq \M $ such that $X_{\alpha}
\rightarrow Y$ with respect to $\tau_p;$ hence

\begin{equation} \label{2.7} X_\alpha P_j \rightarrow Y P_j \in L^p(\vp) \quad \mbox{in} \norm {\cdot}_{p,\vp}.\end{equation}
But for each central $\vp-$finite projection $P$ we have \begin{equation}\label{infinita}\vp(P)=\vp(P\sum_{j\in J} P_j)=\sum_{j\in J} \vp(PP_j). \end{equation}

By our assumption, for any $l, m\in J$, we may pick  $U\in\M$ so
that $U^*U=P_{l}$ and $UU^*=P_{m},$ hence,
$$ \vp(PP_{l})=\vp( PU^*U )=\vp(UPU^*)=\vp(PUU^*)=\vp(PP_{m}).$$
So, all terms on the right hand side of \eqref{infinita} are equal
and since the above series converges, only a finite number of them
can be nonzero. Thus, for some $s \in {\mb N}$ we may write
$J=\{1,\ldots,s\}$ and then
 \begin{equation} P=P\sum_{j\in J} P_j=P\sum_{j=1}^{s} P_j=\sum_{j=1}^{s}P P_j
\end{equation}
and hence
\begin{equation} YP=\sum_{j=1}^{s}Y P P_j=\sum_{j=1}^{s}Y P_jP\in L^p(\vp).
\end{equation}
Therefore, if $Y\in \widetilde{\M}^{\tau_p}$, for each central $\vp-$finite projection $P$, we have $YP\in L^p(\vp).$ Hence
$L^p_{\rm loc}(\vp)\supseteq\widetilde{\M}^{\tau_p} .$

\end{proof}

\berem
In general it is not guaranteed  that a von Neumann algebra possesses an orthogonal family $\{P_j\}_{j\in J}$ of mutually equivalent finite central projection such that $\bigvee_{j\in J} P_j=\id, $ but, if this is the case, then $L^p_{\rm loc}(\vp)=\widetilde{\M}^{\tau_p}.$
\enrem

\begin{thm}
Let $\M$ be a von Neumann algebra on a Hilbert space $\HH$ and
$\vp$ a normal faithful semifinite trace on $\M_+$. Then
$(\widetilde{\M}^{\tau_p}, \M)$ is a {\it  locally convex quasi
C*-algebra} with respect to $\tau_p,$ consisting  of measurable
operators.
\end{thm}
\begin{proof}

The topology  $\tau_p$ satisfies the properties (T$_1$)-(T$_4$). We just prove here  (T$_3$)-(T$_4$).

\begin{itemize}\label{loccov}
\item
(T$_3$) For each $\lambda\in J$,
\[
q_{\lambda}(XY)=\norm{P_\lambda XY}_{p,\vp} \leq \| X\|
\norm{P_{\lambda} Y}_{p,\vp} = \| X\| q_{\lambda}(Y), \quad
\forall X, Y \in \M;\]
\item (T$_4$) The set $\mathcal{U}(\M)_+ := \{ X \in (\M)_+ : \| X \| \leq 1 \}$ is $\tau_p$-closed. To see this consider a net
$\{F_\alpha\}$  in $\mathcal{U}(\M)_+$ : $F_\alpha \to F$ in the
topology $\tau_p$, then for each ${j\in J}$, we have
$\norm{(F_\alpha -F)P_j}_{p, \vp} \rightarrow 0.$ By assumption on $P_j$,
the trace $\vp$ is a normal faithful finite trace on the von
Neumann algebra on  $ P_j \M_+.$ Then (see \cite{bttjmp}) $(L^p(\vp),
P_j \M )$ is a CQ*-algebra. Therefore, using (T$_4$) for
$(L^p(\vp), P_j \M ),$ we have $ FP_j\in \mathcal{U}(P_j \M)_+$
for each ${j\in J}.$ This, by definition, implies that $F\in\M.$
Indeed, for every $$h=\sum_{j\in J}P_j h \in\HH=\bigoplus_{j\in J}
P_j \HH$$ we have
\begin{eqnarray*}\norm{F h}_{\HH}^2 &=& \sum_{j\in J}\norm{F P_j h }^2 \\
 &=& \sum_{j\in J}\norm{FP_jP_j h }^2  \leq  \sum_{j\in J}\norm{P_j h }^2 \\
 &=&\norm{h}_{\HH}^2.
\end{eqnarray*}
Hence  $F\in\mathcal{U}(\M)_+.$
\end{itemize}

\end{proof}
\berem
By  Proposition \ref{prop}, $(L^p_{\rm loc}(\vp), \M)$ itself is a {\it  locally convex quasi C*-algebra} with respect to $\tau_p.$
\enrem

\section{The representation theorems}

Let  $(\X,\Ao)$ be a
locally convex quasi
C*-algebra with a unit $e.$ For each $\Omega\in\mathcal{T}(\X)$, we
define a linear functional $\omega_\Omega$ on $\Ao$ by
$$\omega_\Omega(a):=\Omega(a,e)\quad \quad a\in\Ao.$$
We have
$$\omega_\Omega(a^{\ast}a)=\Omega(a^{\ast}a,e)=\Omega(a,a)=\Omega(a^{\ast},a^{\ast})=\omega_\Omega(aa^{\ast})\geq0.$$
This shows at once that $\omega_\Omega$ is positive and tracial.

By the Gelfand-Naimark theorem each $C^{\ast}$-algebra is
isometrically *-isomorphic to a $C^{\ast}$-algebra of bounded
operators in Hilbert space. This isometric *-isomorphism is called
the {\it universal *-representation}. We denote it by $\pi$.

For every $\Omega\in\mathcal{T}(\X)$ and $a\in\Ao$, we put
$$\vp_\Omega(\pi(a))=\omega_\Omega(a).$$

Then, for each $\Omega\in \mathcal{T}(\X)$, $\vp_\Omega$
is a positive bounded linear functional on the operator algebra
$\pi(\Ao).$

Clearly,
$$\vp_\Omega(\pi(a))=\omega_\Omega(a)=\Omega(a,e).$$
{  Using
 the fact that the family $\{p_\lambda\}$ is directed, we conclude that  there exist $\gamma>0$ and $\lambda \in \Lambda$ such that
$$\mid\vp_\Omega(\pi(a))\mid=\mid\omega_\Omega(a)\mid=
\mid\Omega(a,e)\mid\leq \gamma^2 p_{\lambda}(ae)p_\lambda(e).$$
Then, by (T$_3$), for some $\lambda'\in \Lambda$

$$\mid\vp_\Omega(\pi(a))\mid \leq \gamma^2 \| a \|_{\scriptscriptstyle 0} p_{\lambda'}( e)^2.$$
}

Thus $\vp_\Omega$ is continuous on $\pi(\Ao).$

By \cite[Vol. 2, Proposition 10.1.1]{kadison}, $\vp_\Omega$ is weakly
continuous and so it extends uniquely to
$\pi(\Ao)^{''}$, by the Hahn-Banach theorem. Moreover,
since $\vp_\Omega$ is a trace on $\pi(\Ao)$, the
extension  $\widetilde{\vp_\Omega}$ is also a trace on  the von Neumann algebra $\mathfrak{M}:=\pi(\Ao)^{''}$  generated by
$\pi(\Ao)$.

Clearly, the set $\mathfrak{N}_\mathcal{T}(\Ao)=\{\widetilde{\vp_\Omega};\,
\Omega\in\mathcal{T}(\X) \}$ is convex.

\bedefi \label{nuovasemi}
The locally convex quasi C*-algebra  $(\X,\Ao)$ is said to be {\em strongly
*-semisimple}  if
\begin{itemize}
\item[(a)] the equality $ \Omega(x,x)=0, $\, for every
$\Omega\in \mathcal{T}(\X)$, implies $x=0$;\label{SLP2}
\item[(b)] the set $\mathfrak{N}_\mathcal{T}(\Ao)$ is $w^{\ast}$-closed.
\end{itemize}
\findefi

\medskip
\begin{rem}
{\em If $(\X,\A_{\scriptscriptstyle 0})$ is a {\em CQ*-algebra}, by \cite[Proposition 4.1]{bttjmp}, (b) is automatically satisfied}.
\end{rem}

\beex Let $\M$ be a von Neumann
algebra and $\vp$ a normal faithful semifinite trace on $\M_+$.
{ Then, as seen in Example \ref{ex_2.3},  if $\vp$ is a finite trace, then $(L^p(\vp), \M)$, with $p\geq 2$, is a *-semisimple  locally convex quasi C*-algebra. The conditions (a) and (b) of Definition \ref{nuovasemi} are satisfied.
 Indeed, in this case, the set $\mathfrak{N}_\mathcal{T}(\M)$ is $w^{\ast}$-closed by \cite[Proposition 4.1]{bttjmp}.
  Therefore $(L^p(\vp), \M)$, with $\vp$ finite, is a {strongly *-semisimple} locally convex quasi C*-algebra .}
 \enex

Let $(\X,\Ao)$ be a locally convex quasi C*-algebra with unit $e$, $\pi$ the universal representation of $\Ao$ and $\M=\pi(\Ao)^{''}$.
Denote by $\norm{f}^\sharp$ the norm of a bounded
functional $f$ on $\M$ and by  $\M^{\sharp}$ the topological dual of $\M$ then,
the norm $\norm{\widetilde{\vp_\Omega}}^\sharp$ of
$\widetilde{\vp_\Omega}$ as a linear functional on $\mathfrak{M}$
equals the norm of $\vp_\Omega$ as a functional on
$\pi(\Ao).$

{ By (iv) of Definition \ref{SLP},
$\norm{\widetilde{\vp_\Omega}}^\sharp={\widetilde\vp_\Omega}(\pi(e))=\Omega(e,e)\leq 1.
$}

Hence, if (b) of Definition \ref{nuovasemi} is satisfied, the set $\mathfrak{N}_\mathcal{T}(\Ao)$, being a $w^{\ast}$-closed subset of the unit ball of $\M^{\sharp}$,
is $w^{\ast}$-compact.

Let $\mathfrak{EN}_\mathcal{T}(\Ao)$ be
the set of extreme points of $\mathfrak{N}_\mathcal{T}(\Ao)$;
then $\mathfrak{N}_\mathcal{T}(\Ao)$ coincides with
$w^{\ast}$-closure of the convex hull of
$\mathfrak{EN}_\mathcal{T}(\Ao).$
Let
 $\pi:\Ao \hookrightarrow \BH$ be its universal representation.

Then $\mathfrak{EN}_\mathcal{T}(\Ao)$ is a family of normal finite traces on the von Neumann
algebra $\mathfrak{M}.$

 We put $\mathcal{F}:= \{\Omega\in\mathcal{T}(\X);\,   \, \widetilde{\vp_\Omega}\in\mathfrak{EN}_\mathcal{T}(\Ao)\}$
 and denote by  $P_\Omega$  the support projection corresponding to the trace $\widetilde{\vp_\Omega}.$
 By \cite[Lemma 3.5]{bttjmp},
 $ \{P_\Omega\}_{\Omega\in\mathcal{F}} $ consists of mutually
orthogonal projections and if $Q:=\bigvee_{\Omega\in \mathcal{F}} {P_\Omega}$ then   $$\mu= \sum_{\widetilde{\vp_\Omega}\in\mathfrak{EN}_\mathcal{T}(\Ao)} \widetilde{\vp_\Omega} $$ is a  normal faithful semifinite trace defined on
the direct sum (see \cite{takesaki} and \cite{S.Triolocirc}) of  von Neumann algebras
 $$Q\mathfrak{M}=\bigoplus_{\Omega\in \mathcal{F}} P_\Omega \mathfrak{M}.$$

\begin{thm}
Let $(\X, \Ao)$ be a strongly  *-semisimple locally convex quasi
C*-algebra with unit $e$, $\pi$ the universal representation of
$\Ao$. Then there exists a monomorphism $$\Phi: x\in\X \rightarrow
\Phi(x):=\widetilde{X}\in\widetilde{{Q\M}}^{\tau_2}$$ with the
following properties:

(i) $\Phi$ extends the isometry $\pi:\Ao \hookrightarrow \BH$
given by the Gelfand-Naimark theorem;

 (ii) $\Phi(x^*)=\Phi(x)^*,\,
$ for every $\,x \in \X;$

 (iii) $\Phi(xy)=\Phi(x)\Phi(y)$, for all
$x,y\in \X$ such that $x \in \Ao$ or $y \in \Ao.$
\end{thm}
\begin{proof}
Let  {  $\{p_\lambda\}_{\lambda \in \Lambda}$ be, as before, the family} of
seminorms defining the topology $\tau$ of $\X.$ For fixed
$x\in\X$, there exists a net $\{a_\alpha;\, \alpha\in \Delta \}$ of
elements of $\Ao$ such that $p_\lambda(a_\alpha-x) \rightarrow 0$
for each ${\lambda \in \Lambda}.$ We put $X_\alpha= \pi(a_\alpha).$

By (iii) of Definition \ref{sesqu}, for every $\Omega\in
\mathcal{T}(\X),$ there exist { $\gamma>0$ and  ${\lambda' \in \Lambda}$}  such that for each $\alpha, \beta\in\Delta$

\begin{eqnarray*} \label{singolo} \|P_\Omega (X_\alpha-X_\beta)\|_{2,\widetilde{\vp_{\Omega}}}
&=&\|P_\Omega(\pi(a_\alpha)-\pi(a_\beta))\|_{2,\widetilde{\vp_{\Omega}}}\\
&=&[\widetilde{\vp_{\Omega}}(|P_\Omega(\pi(a_\alpha)-\pi(a_\beta))|^2)]^{1/2}=\\
&=&[\Omega \left((a_\alpha-a_\beta)^*(a_\alpha-a_\beta),e\right)]^{1/2} \\ &=&  [\Omega(a_\alpha-a_\beta,a_\alpha-a_\beta)]^{1/2}\leq \gamma p_{\lambda'}(a_\alpha-a_\beta)
\underset{\alpha, \beta}\longrightarrow 0.
\end{eqnarray*}

 Let $ \widetilde{X_\Omega}$ be the
$\|\cdot\|_{2,\widetilde{\vp_{\Omega}}}$-limit
of the net $(P_\Omega{X_\alpha})$ in $L^2(\widetilde{\vp_\Omega}).$
Clearly $ \widetilde{X_\Omega}=P_{\Omega}\widetilde{X_\Omega}$

We define $$\Phi(x):=\sum_{\Omega\in \mathcal{F}} P_{\Omega}
\widetilde{X_\Omega}=:\widetilde{X}.$$ Clearly  $\widetilde{X}\in
\widetilde{{Q\M}}^{\tau_2}$.

It is easy to see that the map $x \ni \X \mapsto
\widetilde{X}\in\widetilde{Q\M}^{\tau_2}$ is well defined and injective.
Indeed, if $a_\alpha \rightarrow 0  $, { there exist $\gamma >0$ }and ${\lambda' \in \Lambda}$  such that

\begin{eqnarray*} \label{singolo2} \|P_\Omega X_\alpha\|_{2,\widetilde{\vp_{\Omega}}}
&=&\|P_\Omega\pi(a_\alpha)\|_{2,\widetilde{\vp_{\Omega}}}\\
&=&[\widetilde{\vp_{\Omega}}(|P_\Omega(\pi(a_\alpha)|^2)]^{1/2}=\\
&=&[\Omega \left(a_\alpha^*a_\alpha,e\right)]^{1/2} \\ &=&  [\Omega(a_\alpha,a_\alpha)]^{1/2}\leq \gamma p_{\lambda'}(a_\alpha)
\rightarrow 0.
\end{eqnarray*}
Thus $P_\Omega (X_\alpha)= 0$ for every $\Omega\in
\mathcal{T}(\X),$  then $\widetilde{X}=0.$
Moreover if  $P_{\Omega} \widetilde{X}=0$, for each $\Omega\in\mathcal{F},$ then
 $\Omega(x,x)=0$ for every $\Omega\in\mathcal{F}.$
{  Since every $\Omega \in {\mc T }(\X)$ is a w*-limit of of convex combinations of elements of ${\mc F}$, we get}
$\Omega(x,x)=0$ for every $\Omega\in
\mathcal{T}(\X).$ Therefore, by assumption, $x=0.$

\end{proof}
\berem

In the same way we can prove that

\begin{itemize}
\item
If $(\X, \Ao)$ is a strongly  *-semisimple locally convex quasi
C*-algebra and there exists a faithful $\Omega\in\mathcal{T}(\X)$
(i.e., the equality $ \Omega(x,x)=0, $\, implies $x=0$) then there
exists a monomorphism $$\Phi: x\in\X \rightarrow
\Phi(x):=\widetilde{X}\in L^2(\widetilde{\vp_{\Omega}})$$ with the
following properties:

(i) $\Phi$ extends the isometry $\pi:\Ao \hookrightarrow \BH$
given by the Gelfand-Naimark theorem;

(ii) $\Phi(x^*)=\Phi(x)^*, \, $ for every $x \in \X$

(iii) $\Phi(xy)=\Phi(x)\Phi(y)$, for all $x,y\in \X$ such that $x
\in \Ao$ or $y \in \Ao.$

\item If the semifinite von Neumann algebra $\pi(\Ao)''$ admits an orthogonal family of mutually equivalent projections  $\{P'_i; \,i \in I\} $ such that $\sum_{i\in I} P'_i=\id$,
then is easy to see that the map $x \in \X \to
\widetilde{X}\in L^2_{loc}(\tau)$ is a monomorphism.
\end{itemize}
\enrem

\medskip
\noindent{\bf Acknowledgement.\;} The authors wish to express their gratitude to the referee for pointing out several inaccuracies in a previous version of the paper and for his/her fruitful comments.

\medskip

\bibliographystyle{amsplain}

\end{document}